\documentclass{ifacconf}

\usepackage{graphicx}      
\usepackage{amssymb}
\usepackage{amsmath}
\usepackage{color}
\usepackage{comment}
\usepackage{enumitem}
\usepackage[numbers]{natbib}        

\newcommand{\R}{\mathbb{R}}
\newcommand{\be}{\begin{equation}}
\newcommand{\ee}{\end{equation}}
\newcommand{\ba}{\begin{aligned}}
\newcommand{\ea}{\end{aligned}}

\newcommand{\lba}{\left[ \begin{array}}
\newcommand{\ear}{\end{array} \right]}

\newcommand{\haone}{\hat{a}_1(\kappa)}
\newcommand{\hatwo}{\hat{a}_2(\kappa)}
\newcommand{\hpzero}{\hat{p}_0(\kappa)}
\newcommand{\hpone}{\hat{p}_1(\kappa)}
\newcommand{\hptwo}{\hat{p}_2(\kappa)}
\newcommand{\hqzero}{\hat{q}_0(\kappa)}
\newcommand{\hqtwo}{\hat{q}_2(\kappa)}

\usepackage{bm}

\begin{document}
\begin{frontmatter}

\title{Conditions for Complete Decentralization of the Linear Quadratic Regulator} 


\author[First]{Addie McCurdy} 
\author[First]{Isabel Collins} 
\author[First]{Emily Jensen}

\address[First]{University of Colorado Boulder, 
   Boulder, CO 80309 USA (e-mails: \{addie.mccurdy, isabel.collins, ejensen\} @colorado.edu).}

\begin{abstract}                
An unconstrained optimal control policy is \emph{completely decentralized} if computing actuation for each subsystem only requires information directly available to its own subcontroller. Parameters that admit a completely decentralized optimal controller have been characterized in a variety of systems, but attempts to physically explain the phenomenon have been limited. As a step toward a general characterization of complete decentralization, this paper presents conditions for complete decentralization of Linear Quadratic Regulators for several simple cases and physically interprets these conditions with illustrative examples. These simple cases are then leveraged to characterize complete decentralization of more complex systems.  
\end{abstract}

\begin{keyword}Decentralized control, Controller constraints and structure, Linear systems, Optimal control theory, Distributed control and estimation
\end{keyword}

\end{frontmatter}

\section{Introduction}
 As refined sensing and actuation technologies allow for access to large amounts of data, real-time control of increasingly complex systems can become computationally infeasible. Optimal control often requires access to complete information, which can lead to a communication bottleneck, especially in a distributed framework where all subcontrollers need to share information with each other. 
 
One method to address this bottleneck is to incorporate a sparsity constraint or penalty to the optimal controller design problem so as to guarantee a sufficiently decentralized controller [\cite{moteeDecentralizedOptimalControl2008, fardad2011sparsity}]. Because sparsity constrained problems are not convex in general, solutions can only be found for systems with special properties, such as quadratic invariance or funnel causality,  which allow for reformulation into convex problems \citep{arbelaizOptimalStructuredControllers2021, lessardConvexityDecentralizedController2016,bamieh2005convex}. However, it has been demonstrated that in system design, parameters can be chosen so that the \textit{unconstrained} optimal controller naturally has desirable communication properties \citep{jensen2020localization, bamieh2002distributed,arbelaiz2024optimal}. As an alternative to constrained optimization, this paper analytically finds unconstrained optimal controllers, then determines conditions on control parameters and system dynamics that ensure the optimal controller has desired sparsity properties. This yields results applicable for use in controller/plant co-design. 
 
  In the most extreme case,  parameter choices made in a Linear Quadratic Regulator (LQR) {design problem} allow for the {unconstrained} optimal control policy to be \textbf{completely decentralized}, which means that the optimal control for a given subsystem can be computed only using knowledge of state elements that are directly available to its subcontroller.  If each subsystem is outfitted with its own computation unit and actuator, then a completely decentralized control policy requires no communication between computation units, thus reducing computational complexity.

Complete decentralization has been observed in multiple fundamentally different cases. In the spatially invariant setting with a continuous spatial domain, \cite{arbelaiz2020distributed} and \cite{arbelaiz2022information} describe conditions for complete decentralization of the Kalman filter for the heat and wave equations, respectively, over the real line.  These examples describe spatial properties of the optimal estimators as a function of system parameters, but do not consider the resulting steady state covariance and its relationship to estimator locality. In the spatially invariant setting with a finite or countably infinite number of spatial sites, \cite{mccurdy2025complete} derived a similar condition for decentralization of the Linear Quadratic Gaussian (LQG) controller for the discrete wave equation over the unit circle. In addition, this work computed the optimal $H_2$ norm of the resulting optimal controllers for varying system parameters, and showed that parameters that completely decentralize the system do not coincide with a higher optimal $H_2$ norm, demonstrating that there is not an inherent tradeoff between spatial locality and performance of an optimal controller. Even in the spatially variant case, LQR for the wave equation on an interval with homogeneous Dirichlet boundary conditions can be completely decentralized using methods outlined in \cite{epperlein2016spatially}.  

The previous examples illustrate that the complete decentralization of optimal control policies is not limited to a certain class of systems. Rather, decentralized optimal controllers can be derived for systems that are spatially variant or invariant, with infinite or finite dimensional dynamics. However, thus far complete decentralization has only been characterized for specific examples or cases. General conditions on system structures that admit complete decentralization have not yet been determined. Furthermore,  existing conditions have not been interpreted in a way that physically explains when an optimal controller is completely decentralized.  It is rather counterintuitive that in these cases, better performance cannot be gained with access to more information. Thus, further analysis is needed to understand how system properties allow for this phenomenon.

In this paper we take a step toward characterization of general conditions for complete decentralization by examining simple examples. By fully understanding base cases, we hope to gain insight into the cause of complete decentralization for larger, more complex systems. We first derive a condition for complete decentralization of LQR for a subset of  $2\times2$ discrete systems. Next, we fully characterize decentralization for a class of spatially invariant discrete systems of any size, which we then specialize to $2\times2$ spatially invariant systems. In each case, we physically interpret the conditions we derive and present illustrative examples. Finally, we show how more complex systems can be completely decentralized by reducing them to the previous simple cases. 

\subsection{Notation}
In this paper we work with vector valued signals that continuously evolve in time, e.g. $x(t)\in\R^n.$
Hats denote the unitary spatial Discrete Fourier Transform (DFT) of signals: 
\begin{equation}
    \hat{x}_\kappa(t)=(Fx)(\kappa,t):=\frac{1}{\sqrt{n}}\sum_{j=1}^{n-1}x_j(t)\exp{\left(-2\pi i\frac{\kappa}{n}j\right)}
\end{equation}
where $\kappa\in\{0,\ldots,n-1\}$ is the spatial frequency, { $j = \sqrt{-1}$,} and $x_n(t)$ denotes the $n^{\rm th}$ entry of the vector $x(t) \in \mathbb{R}^n$. We write $\hat{x}=[\hat{x}_0,\ldots,\hat{x}_{n-1}]^\top\in\R^n$. The DFT can be written as multiplication by a matrix $F$ so that $\hat{x}=Fx.$ 

An {$n \times n$} \textit{circulant matrix} is completely determined by its first row, with the $j^{th}$ row equal to the first row shifted right $j-1$ entries. The DFT diagonalizes circulant matrices, i.e. if $M$ is an $n\times n$ circulant matrix, then $\hat{M}=FMF^{-1}$ where \begin{equation} \label{eq:eigs}\hat{M}_{\kappa,\iota}=\begin{cases}
    \hat{m}(\kappa) & \kappa=\iota\\
    0 & \kappa\neq\iota
\end{cases}\end{equation} and $\hat{m}(\kappa), \kappa=1,\ldots, n$ is the $\kappa^{th}$ eigenvalue of $M$.

\section{Problem Setup}
Throughout this paper we consider dynamics of the form \begin{equation}\label{eq:statespace}
    \dot{x}=Ax+Bu
\end{equation} where $x(t)\in\R^n$, $u(t)\in \mathbb{R}^m$ are the state and input, respectively, $A$ is the state transition matrix, and $B$ is the input multiplier. The LQR design problem seeks to minimize the cost 
\begin{equation}\label{eq:LQR}
    J= \int_0^{\infty} x^\top Qx+u^\top Ru \;dt\\
    \end{equation}
subject to dynamics \eqref{eq:statespace} for some positive definite matrices $Q,R$. We will always assume that $(A,B)$ is stabilizable and $(A,Q)$ is detectable. Then, it is well known that the LQR solution is given by $u^*=Kx$ where $K=R^{-1}B^\top P$ and $P$ is the unique positive definite solution to the Algebraic Riccati Equation (ARE) 
\begin{equation}\label{eq:ARE}
     A^{\top} P + P A - PB R^{-1}B^{\top} P + Q = 0.
\end{equation}
We assume that each element of the input $u_i$ is computed by its own subcontroller $Z_i$. We use $N_i$ to denote the subsystem of states that $Z_i$ has direct access to, i.e. without any communication with other subcontrollers.  
\textbf{Complete decentralization} occurs when the optimal control $u^*$ is of the form $u^*_i=\sum _{s\in N_i}k_s s$ for some scalars $k_s$ for all $i=1,\ldots,m$. This means that each actuation signal can be computed with no information from other subcontrollers, as shown in Figure \ref{fig:decentralization}. In the case that $m=n$, complete decentralization is equivalent to $K$ being diagonal. In this work, we derive conditions on system parameters that allow for the unconstrained optimal controller to be completely decentralized. 

\begin{figure}
\begin{center} 
\includegraphics[width=7cm]{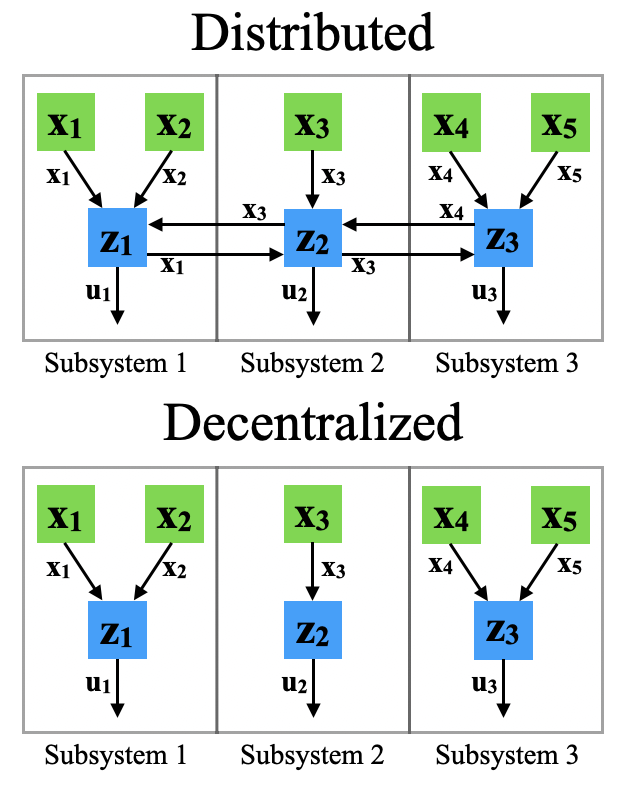}    
\caption{Visualization of distributed (top) and decentralized (bottom) control policies for three subsystems with $N_1=\{x_1,x_2\}, N_2=\{x_3\}, N_3=\{x_4,x_5\}$. In the distributed policy, the subcontrollers $Z_i$ are shown relaying state information.} 
\label{fig:decentralization}
\end{center}
\end{figure}

\subsection{Trivial Case: $A,B,Q,R$ diagonal}
 When $A,B,Q,R$ are all diagonal matrices, LQR is always completely decentralized.
This follows directly from the fact that the solution $P$ to the ARE \eqref{eq:ARE}, and thereby $K$, will be diagonal. This is intuitive, since a diagonal system has no coupling between state elements, so the optimal control policy will likewise be decoupled. A more interesting case is when state elements are dynamically coupled, yet the system can still be completely decentralized. We will explore examples of this in the following sections.

\section{The $2\times2$ Linear Quadratic Regulator with diagonal $B,R,Q$.}
In this section we focus on LQR for a $2\times 2$ system, $(x,u\in\mathbb{R}^2)$, where all matrices are diagonal except $A$ which couples the dynamics. In particular, diagonal $Q,R$ matrices represent a choice of \emph{decoupled} cost. We assume without loss of generality that $B =I$ (achieved by rescaling $u$). Now the dynamics of interest are of the form 
\begin{equation} \label{2x2gen}
    \dot{x} =  \lba{cc} a_0 & a_1 \\ a_{-1} & a_2 \ear x + I u,
\end{equation}
with LQR weights 
\begin{equation} \label{2x2cost}
Q=\begin{bmatrix}
    q_0 & 0 \\ 0 & q_2
\end{bmatrix}, R=\begin{bmatrix}
    1/\gamma_0 & 0 \\ 0 & 1/\gamma_2
\end{bmatrix}
\end{equation}

\begin{thm}\label{thm:2x2simplecondition}
The following conditions on the state transition matrix of  \eqref{2x2gen} are sufficient for complete decentralization
of the LQR problem for some choice of decoupled cost \eqref{2x2cost}:
\begin{enumerate}[label=(\roman*)] 

    \item\label{cond:offdiag}
    $a_1$ and $a_{-1}$ have opposite signs.

    \item\label{cond:diag}
    $a_0$ and $a_2$ have the same sign.
\end{enumerate}
The corresponding choice of cost \eqref{2x2cost} that decentralizes the LQR satisfies:
\begin{enumerate}[start=3,label=(\roman*)]
    \item\label{cond:q}
    $
        \frac{q_0}{q_2}
        =
        \frac{-a_0 a_{-1}}{a_1 a_2}
        \label{eq:qcondition}
    $

        \item\label{cond:gamma}
    $
        \frac{\gamma_0}{\gamma_2}
        =
        \frac{a_1^2}{a_{-1}^2}\frac{q_0}{q_2}
        \label{eq:gammacondition}.
    $

\end{enumerate}
\end{thm}

\subsection{Performance analysis}\label{sec:perf}
One metric of a controller's performance is the system's closed loop $H_2$ norm. For a system of the form \eqref{eq:statespace}, the $H_2$ norm of the system with an LQR controller is equal to the value of J in \eqref{eq:LQR} [\cite{doyle1989state}]. While each controller is optimal for the given system parameters, optimal  controller performance in terms of the closed loop $H_2$ norm will vary from system to system as different parameters are chosen.

To demonstrate that a completely decentralized controller does not necessarily result in low performance, consider as an example the system: 
\begin{equation} \label{2x2gen_forperf}
    \dot{x} =  \lba{cc} 1 & 1 \\ -1 & 1 \ear x + I u,
\end{equation}
 
\begin{equation} \label{2x2cost_forperf}
    Q=\lba{cc} q_0 & 0 \\ 0 & 1 \ear , R=\lba{cc} 1 & 0 \\ 0 & \tfrac{1}{\gamma_2}\ear
\end{equation}

\begin{figure}
\begin{center} 
\includegraphics[width=8.4cm]{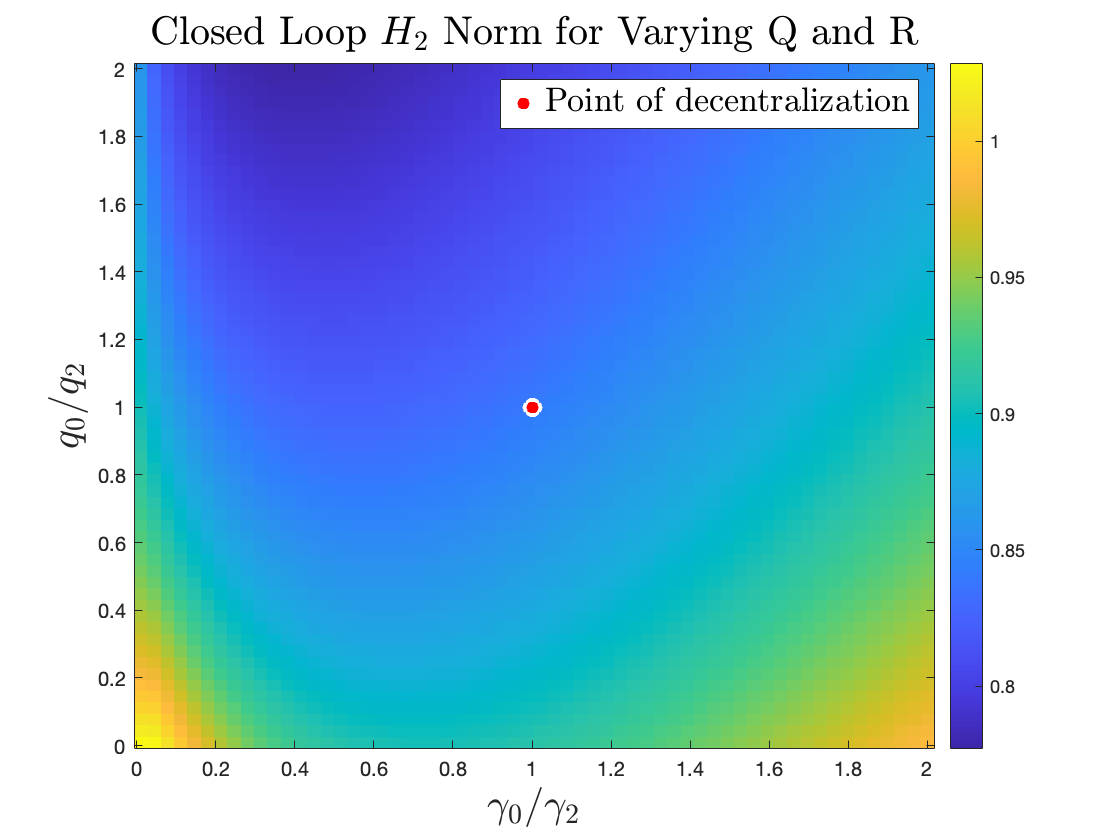}    
\caption{Heat plot of closed loop $H_2$ norm with varying choices of $Q$ and $R$. The point in red marks the values of $q0/q2$ and $\gamma_0/\gamma_2$ that result in decentralization.} 
\label{fig:QvsR}
\end{center}
\end{figure}

 Figure \ref{fig:QvsR} shows the closed loop $H_2$ norm of this system in feedback with the optimal LQR controller for various choices of $\frac{\gamma_0}{\gamma_2}$ and $\frac{q_0}{q_2}$. In this case, the parameter choices that give a completely decentralized controller (shown by the red dot in the figure) result in neither the highest nor the lowest closed loop $H_2$ norm.

Because there is not a direct tradeoff between decentralization and performance, it may be possible to strategically choose parameters in system design to decrease the optimal cost of the decentralized controller (as discussed in \cite{mccurdy2025complete}).  Figure \ref{fig:QvsA} depicts a heat map of the closed loop $H_2$ norm of the system given by \ref{2x2gen_forperf} as $q_0$ and $a_2$ are varied and $\gamma_2$ is chosen to satisfy \ref{cond:gamma}. The values of $q_0/q_2$ and $a_2/a_0$ that satisfy \ref{cond:q} give a curve of decentralization, shown in black. The curve crosses values of both high and low $H_2$ norm indicating that, in this case, the performance of the decentralized controller can be improved by choice of $\frac{a_2}{a_0}$. 

\begin{figure}
\begin{center} 
\includegraphics[width=8.4cm]{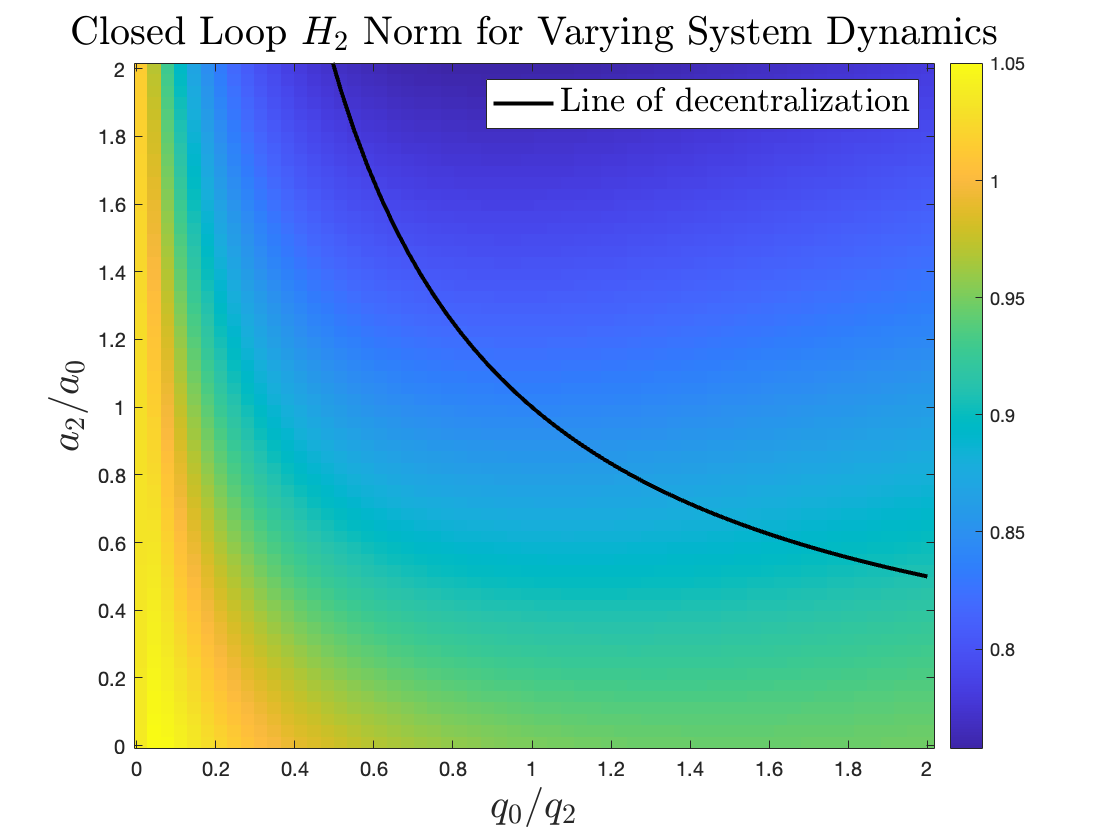}    
\caption{Heat plot of closed loop $H_2$ norm with varying choices of $Q$ and ${a_2}/{a_0}$. The curve of decentralization is shown in black.} 
\label{fig:QvsA}
\end{center}
\end{figure}

\subsection{Example: aquarium population dynamics}

An example of a physical system described by Theorem \ref{thm:2x2simplecondition} is a fish tank into which fish can be freely added and removed as depicted in Figure \ref{fig:diagram}. The states represent two species whose interactions with each other are described by the off-diagonal terms $a_1, a_{-1}.$ The diagonal terms $a_0, a_2$ represent each species' interactions with itself, and positive signs indicate that breeding benefits outweigh resource competition. Note that $a_1, a_{-1}$ with the same sign indicates a symbiotic relationship while \emph{opposite signs indicates a predator/prey, or a ``competitive" relationship}. Thus, the fact that conditions \ref{cond:offdiag}-\ref{cond:diag} of Theorem~\ref{thm:2x2simplecondition} are sufficient to decentralize an LQR for a \emph{decoupled} cost of the form \eqref{2x2cost} can be interpreted as follows:

~\hspace{2mm} \emph{A competitive relationship in a 2-state system with proportional self-dynamics enables decentralization of LQR for a decoupled cost.}

A system whose dynamics satisfy the decentralization requirements \ref{cond:offdiag}-\ref{cond:diag} given in Theorem \ref{thm:2x2simplecondition} is described by the following linearized predator prey model with logistic growth evaluated at the coexistence equilibrium [\cite{savadogo2021mathematical}]: 
\begin{equation}  \label{eq:pop_dynamics}
    \dot{\delta} = J^\star \delta + I u,
\end{equation}

\begin{equation}
\begin{aligned}
J^\star &=
\begin{pmatrix}
-\dfrac{r_1 r_2 (r_1 - b k_2)}{\sigma} 
& -\dfrac{b k_1 r_2 (r_1 - b k_2)}{\sigma} \\[15pt]
\dfrac{b e k_2 r_1 (r_2 + b e k_1)}{\sigma} 
& -\dfrac{r_1 r_2 (r_2 + b e k_1)}{\sigma}
\end{pmatrix}, \\[12pt]
&\sigma = e k_1 k_2 b^2 + r_1 r_2.
\end{aligned}
\end{equation}

where 

\begin{itemize}
    \item $\delta_1$ and $\delta_2$ represent the prey and predator populations deviation from the equilibrium point, respectively.
    \item $r$, $k$, $b$, and $e$ are positive constants that represent intrinsic growth rate, carrying capacity of the environment, predation rate per unit of time, and conversion rate, respectively.
    \item $u$ represents the number of species added/removed from the system per unit time
\end{itemize}


We further interpret conditions \ref{cond:offdiag}-\ref{cond:diag} of Theorem~\ref{thm:2x2simplecondition} as follows. A completely decentralized controller can be thought of as a set of ``information-greedy" controllers, each working to return its own state to the equilibrium, regardless of the condition of the other states and without leveraging information about the other states. Then the conditions under which the ARE returns a decentralized controller are the system dynamics and 
cost parameters for which this information-greedy controller is optimal. 

The competitive relationship described by conditions \ref{cond:offdiag} and \ref{cond:diag} from Theorem \ref{thm:2x2simplecondition} indicate that an over-correction in the number of predators causes a corresponding decrease in the number of prey, which in turn, causes a decrease in the number of predators, helping to correct the overshoot. In fact, looking carefully, we see that this condition causes \textit{any} over- or under-correction by \textit{either} controller to be dampened by the reaction of the other species. This makes sense as a condition of decentralization for a decoupled cost because one can now apply an input to one species without worrying that this input will be sub-optimal due to effects from coupling 
for certain values of the population of the other species.

Conditions \ref{cond:q} and \ref{cond:gamma} give the choices of LQR parameters that, given a system with the appropriate dynamics, will result in a decentralized LQR. For the linearized system 
\eqref{eq:pop_dynamics}, condition \ref{cond:q} gives 
   $ \frac{q_0}{q_2} = \frac{ek_2r_1}{k_1r_2}.$
 
To interpret this, consider\[q_0 = \frac{ek_2}{r_2}, ~~ q_2 = \frac{k_1}{r_1}.\] This choice corresponds to selecting a weight on state 1 that depends only on the parameters governing the dynamics of state 2 ($k_2,r_2$) and the coupling ($e)$. This suggests that the information about state 2 through state 1's penalty, $q_0$, may ``compensate" for the lack of information collected by state 1 about state 2 in the optimal LQR.

\begin{figure}
\begin{center} 
\includegraphics[width=7cm]{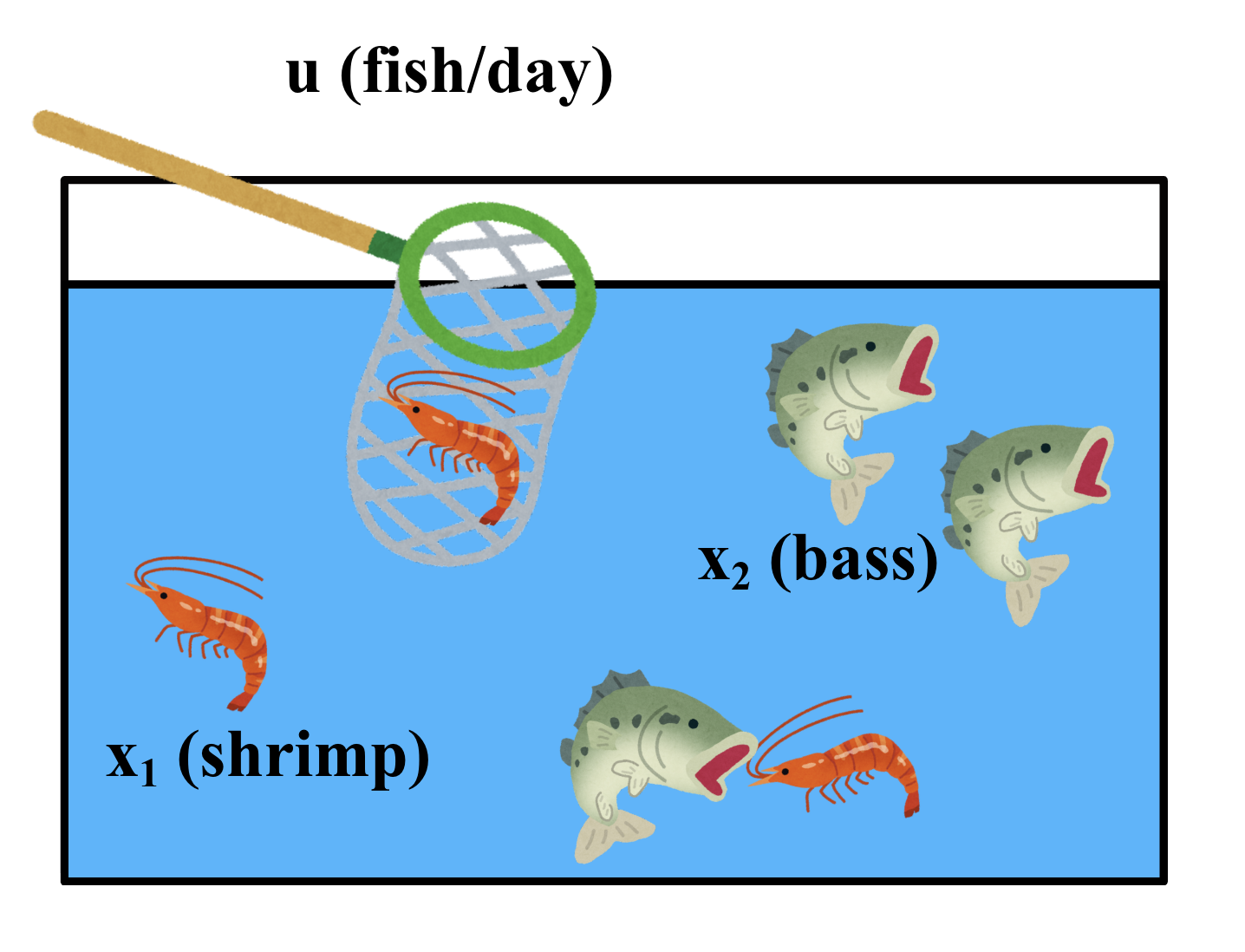}    
\caption{Diagram for physical example of $2\times 2$ system with bass as the predators and shrimp as the prey.} 
\label{fig:diagram}
\end{center}
\end{figure}

\section{Spatially Invariant systems}\label{sec:SIS}
Spatial invariance is the property that system dynamics are invariant with respect to spatial translations. In the discrete setting with finite spatial extent, this is equivalent to all system matrices being (block) circulant. Since the DFT diagonalizes circulant matrices, the LQR problem can be decoupled over spatial frequency in spatially invariant settings (see \cite{bamieh2002distributed,mccurdy2025complete} for details). We derive the following result for a square spatially invariant system of arbitrary size:
\begin{thm}\label{thm:nbynSIcondition}
    Let $A,B,Q,R\in\R^{n\times n}$ be circulant. Then LQR is completely decentralized if and only if there exists a constant $c$ such that $\forall \kappa =0,\ldots,n-1$, 
    \begin{equation}
        \label{eq:SIcondition}
    0 =   c^2 - 2 c \frac{\hat{a}(\kappa)}{\hat{b}(\kappa)} - \frac{\hat{q}(\kappa)}{\hat{r}(\kappa)}
    \end{equation} where $\hat{a}(\kappa), \hat{b}(\kappa), \hat{q}(\kappa), \hat{r}(\kappa)$ are defined as in \eqref{eq:eigs}.\footnote{In this paper, we restricted to systems that are described by ODEs, this result easily generalizes to spatially invariant PDE systems, which can be decoupled into families of finite-dimensional LQR design problems parameterized by frequency. For example, an LQR problem formulated over an $L^2$ space with an infinite dimensional spatial domain can be decentralized using the same condition \eqref{eq:SIcondition} where the hats now denote the eigenvalues of the  infinite dimensional operators $A,B,Q,R:L_2\to L_2$. } The decentralized LQR gain is given by $K=c\cdot I$.
\end{thm}
\begin{pf}
    See Appendix \ref{app:nbynSIcondition}.
\end{pf}
Note that only one constant $c$ is needed to satisfy \eqref{eq:SIcondition} for complete decentralization, but there may be many values of $c$ that work, corresponding to different parameter choices that alter $\hat{a},\hat{b},\hat{q},\hat{r}$. Thus, choice of $c$ alters optimal controller performance in a similar way as described in Section \ref{sec:perf}.

\subsection{Example: discrete diffusion on unit circle}

Let $x\in\R^n$ represent the states of $n$ subsystems evenly distributed in a circle with dynamics \begin{equation} \label{eq:discrete_diffusion}
    \dot{x}=\mathbb{D}^2x+u
\end{equation} where  $\mathbb{D}^2$ is the $n\times n$ circulant matrix with first row $\frac{1}{\Delta^2}\begin{bmatrix}
    -2&1&0&\ldots&0&1
\end{bmatrix}$ and $\Delta$ is the distance between subsystems. These dynamics appear in many physical contexts, such as diffusion modeled as a discretized second order spatial derivative with periodic boundary conditions \citep{recktenwald2004finite}, or in consensus of opinion dynamics \citep{grabisch2020survey}.
Note that in contrast to the predator-prey model, the dynamics \eqref{eq:discrete_diffusion} can be viewed as \emph{cooperative} (rather than competitive) in the sense that the off-diagonal terms of the state transition matrix have the same sign.

Since the system is circulant, we have decoupling over {spatial} frequency in the Fourier domain: \begin{equation}
    \dot{\hat{x}}_\kappa=\hat{D}_\kappa\hat{x}_\kappa+\hat{u}_\kappa
\end{equation} where $\hat{D}_k$ are the eigenvalues of $\mathbb{D}^2$.
Then the condition for decentralization as given in \eqref{eq:SIcondition} is
\begin{equation}\label{eq:heatcodition} \exists c \text{ \;\;s.t. \;\;}
0=c^2-2c\hat{D}_k-\frac{\hat{q}(\kappa)}{\hat{r}(\kappa)}.\end{equation} For $c=1$, \eqref{eq:heatcodition} is satisfied if  $\hat{r}(\kappa)=1$ ($R=I)$, and $\hat{q}(\kappa)=1-2\hat{D}_\kappa$ ($Q=I-2\mathbb{D}^2)$. In the discretization interpretation, we can interpret this choice of $Q$ as including a penalty on the spatial derivative of the state.

To see why, let $\mathbb{D}$ be the $n \times n$ circulant matrix with first row $\frac{1}{\Delta}\begin{bmatrix}
    -1&1&0&\ldots&0&0
\end{bmatrix}$, which is the forward discretization of the first derivative in space. Then $-\mathbb{D}^2=\mathbb{D}^\top\mathbb{D}$, so that \[x^\top Qx=x^\top(I-2\mathbb{D}^2)x=x^\top x+2(\mathbb{D}x)^\top\mathbb{D}x\]  
In the consensus dynamics  interpretation, our condition can be considered as including a penalty on the finite difference between state values. Although other choices of $Q$ and $R$ may be selected, note that no choice of \emph{diagonal} $Q$ and \emph{diagonal} $R$ lead to decentralization of LQR, that is no \emph{decoupled} cost leads to decentralization. We interpret this as follows -

~\hspace{2mm} \emph{A necessary condition for the LQR for the ``cooperative" system \eqref{eq:discrete_diffusion} to be decentralized is that the corresponding cost is coupled.}

Thus, for the 2 examples presented (predator-prey dynamics and diffusion-like dynamics), we observe the following: LQR for a ``competitive" system may be completely decentralized with a \emph{decoupled} cost; LQR for a ``cooperative" system may require \emph{coupled} cost for decentralization. Investigating these patterns for a broader class of systems beyond these examples is the subject of ongoing work.

{\it Remark -} In \cite{arbelaiz2020distributed}, a condition for complete decentralization of the Kalman Filter for diffusion over the real line was related to having spatially autocorrelated measurement noise. This mirrors the result we derived here since in the dual estimation problem, $Q^{-1}$  represents the covariance matrix of the measurement noise, which has full correlation when $Q$ is chosen as above.

\section{The $2\times2$ Spatially Invariant system}  Theorem \ref{thm:nbynSIcondition} does not offer much physical intuition as to why satisfying the condition \eqref{eq:SIcondition} over all frequencies causes the optimal controller to be completely localized. In this section, we examine the simplest spatially invariant case in the original spatial coordinates as opposed to over frequency in order to more easily interpret why complete decentralization is possible.

We will use the result in Theorem \ref{thm:nbynSIcondition} specified to a $2\times 2$ system. In the this case,  $A,B,Q,R$ are of the form  \begin{equation}\label{eq:2by2circ}
    A =\begin{bmatrix}
        a_0 & a_1 \\ a_1 & a_0
    \end{bmatrix},B =\begin{bmatrix}
        b_0 & b_1 \\ b_1 & b_0
    \end{bmatrix},Q =\begin{bmatrix}
        q_0 & q_1 \\ q_1 & q_0
    \end{bmatrix},R =\begin{bmatrix}
        r_0 & r_1 \\ r_1 & r_0
    \end{bmatrix}.
\end{equation} 
\begin{cor}\label{cor:2by2circcond}
     Let $A,B,Q,R\in\R^{2\times 2}$ be defined as in \eqref{eq:2by2circ}. Then a sufficient condition for completely decentralized LQR is  
     \begin{equation}\label{eq:bothroots}\frac{a_0-a_1}{a_0+a_1}=\frac{b_0-b_1}{b_0+b_1}\text{ \;\;and\;\; } \frac{q_0-q_1}{q_0+q_1}=\frac{r_0-r_1}{r_0+r_1}. 
\end{equation}
\end{cor}

\begin{pf}
    See Appendix \ref{app:2by2circcond}
\end{pf}
Note that \eqref{eq:bothroots} is satisfied if all matrices are diagonal $(a_1 = b_1 = q_1 = r_1 =0)$ as shown in the
trivial case. 

 We now interpret the condition \eqref{eq:bothroots}. Because of the circulant structure, coefficients with `0' subscript weight mappings from a coordinate to itself, while coefficients with a `1' subscript weight inter-coordinate mappings. Each ratio above is then a relative difference between the influence a  vector element has on its own coordinate, and the influence it has on the opposite coordinate. To satisfy \eqref{eq:bothroots}, this `relative difference in influence' must be balanced between the state transition matrix $A$ and the input multiplier $B$. At the same time,  the relative difference between weights and cross-weights must be balanced for the state weight matrix $Q$ and the input weight matrix $R$. 

In this case, the decentralization condition is decoupled over the system matrices and the LQR weight matrices, which is not the case for the other conditions we have derived. This implies that if the first condition is satisfied, choosing $Q=A$ and $B=R$ will automatically satisfy the second condition, but if the first condition is not satisfied, then \eqref{eq:bothroots} will not be satisfied for any choice of $Q,R$.

\subsection{Example: heat transfer across a wall} 

\begin{figure}
\begin{center} 
\includegraphics[width=8.4cm]{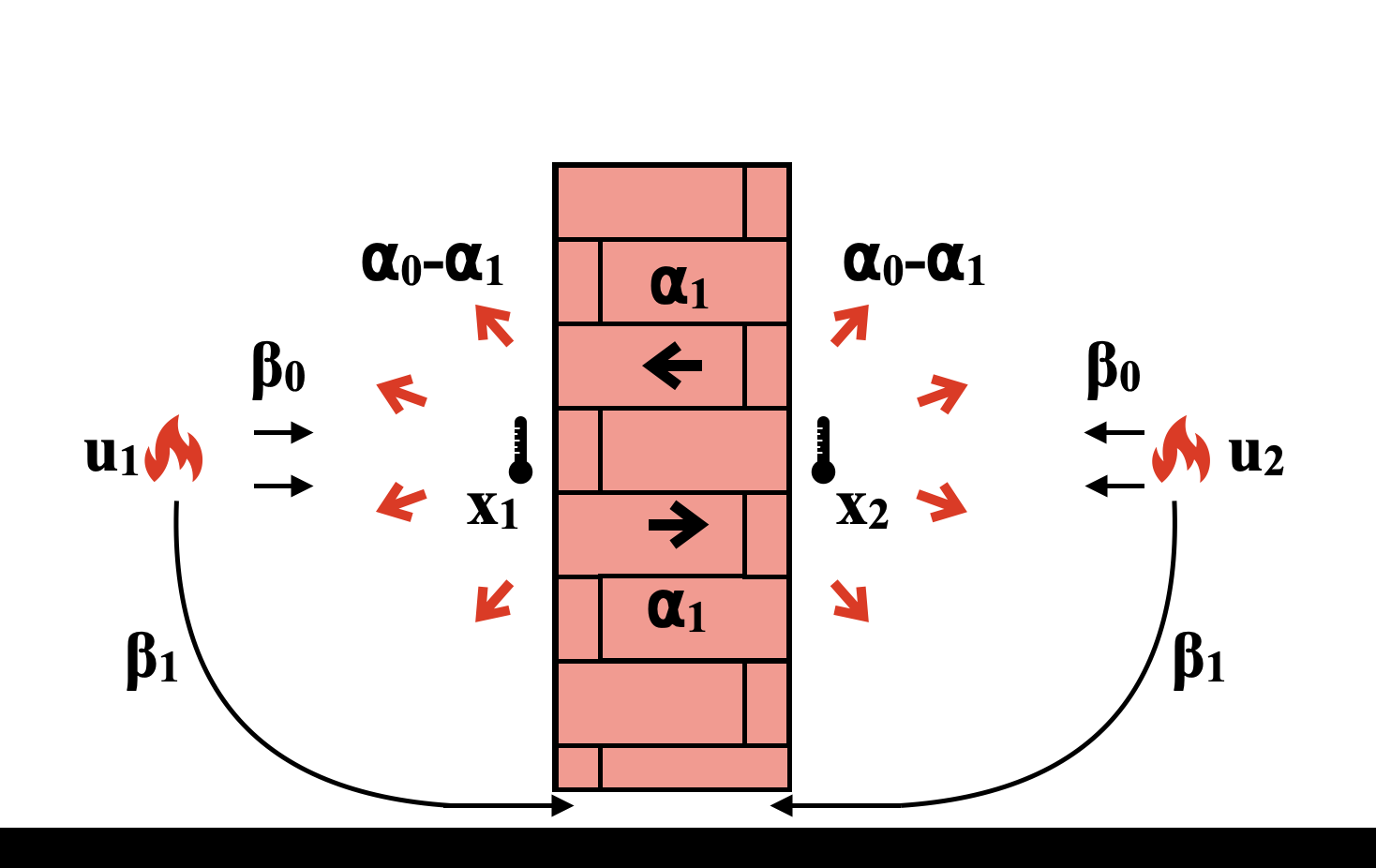}    
\caption{Heat transfer across a wall} 
\label{fig:heat}
\end{center}
\end{figure}
The system in Figure \ref{fig:heat} represents heat transfer across a wall where the chambers on either side of the wall have their own heaters (see \cite{Seem1987} for details). This system can be represented by equations of the form: 

\begin{equation}\label{eq:chambers}
    \dot{x} =  \lba{cc} -\alpha_0 & \alpha_1 \\ \alpha_1 & -\alpha_0 \ear x +\lba{cc} \beta_0 & \beta_1 \\ \beta_1 & \beta_0 \ear u,
\end{equation}

where $x$ is the temperature of each side of the wall, and $u$ represents energy input to each heater. Then $\alpha_0$ represents heat lost to the environment, $\alpha_1$ represents heat transfer between the two sides, $\beta_0$ represents the heat added to each heater's own chamber and $\beta_1$ represents heat added by the heater to the opposite chamber (perhaps from heat transfer through a gap between the floor and the wall). In this example, all system parameters are positive, so it is not possible to satisfy the conditions in Theorem \ref{thm:2x2simplecondition}, but with a dense $B$ matrix, complete decentralization can still be achieved. 

For any diagonal matrices $Q,R\succ 0 $, LQR for \eqref{eq:chambers} is completely decentralized if \begin{equation}
  \frac{\alpha_0-\alpha_1}{\alpha_0+\alpha_1}=\frac{\beta_0-\beta_1}{\beta_0+\beta_1}.
\end{equation}
To interpret this, consider a cold chamber which loses proportionally more heat to the environment than is put into its own room by the chamber so that
\begin{equation}
  \frac{\alpha_0-\alpha_1}{\alpha_0+\alpha_1}>\frac{\beta_0-\beta_1}{\beta_0+\beta_1}.
\end{equation}
The ``information-greedy" decentralized system spends a lot of energy trying to warm its chamber up, but because it loses proportionally more heat to the environment than its heater puts into the room, the greedy heater ``unknowingly" overheats the neighboring chamber, causing suboptimal performance. To avoid this problem and guarantee that the ``information-greedy" strategy will be optimal, it is necessary to design the heating system in agreement with the condition above. This means that the ratio between the heat lost to the environment and the total heat lost must be equal to the ratio between each heater's input to its own chamber and its total heat input.

\section{Extension to Second Order Dynamics}\label{sec:block_circulant}
In the previous sections, we developed intuition to completely decentralize the LQR for simple $2\times 2$ systems, and circulant systems of any size. We will now demonstrate how these cases may be used to inform decentralization of more complex systems. 

Consider a system with dynamics of the form 
\begin{equation}\label{eq:2nd_order}
    \ddot{x} = A_1 x + A_2 \dot{x} + B_0u.
\end{equation} where $x(t),u(t)\in\R^n $ and each subcontroller $Z_i$ has access to the subsystem $N_i=\{x_i,\dot{x}_i\}$. 
These dynamics can be written in state space using matrices with $n\times n$ blocks: 
\begin{equation}\label{eq:dynamicsorder2}
    \frac{d}{dt} \lba{c} x \\ \dot{x} \ear = \lba{cc} 0 & I \\ A_1 & A_2 \ear \lba{c} x \\ \dot{x} \ear + \lba{c} 0 \\ B_0 \ear u.  
\end{equation}

For a cost functional 
\begin{equation}\label{eq:Jorder2}
    J = \int \lba{c} x \\ \dot{x} \ear^{\top} \lba{cc} Q_0 & 0 \\ 0 & Q_2\ear \lba{c} x \\ \dot{x} \ear +  u^{\top}R_0u dt,
\end{equation}
     
with $Q_0,Q_2, R_0\succ0$ the optimal LQR gain is  
\[K = R^{-1}B^\top P=R_0^{-1}\lba{cc} B_0^\top P_1 & B_0^\top P_2 \ear\] where \begin{equation*} A=
    \lba{cc} 0 & I \\ A_1 & A_2 \ear , B=\begin{bmatrix}
        0 \\ B_0
    \end{bmatrix}, Q= \lba{cc} Q_0 & 0 \\ 0 & Q_2\ear , R=R_0
\end{equation*}    and 
$P = \lba{cc} P_0 & P_1 \\ P_1^\top & P_2 \ear$  is the positive definite solution to the appropriate ARE \eqref{eq:ARE}. 
This ARE reduces to 
\begin{equation}\label{eq:bigARE}
    \begin{aligned}
         \lba{cc} A_1^\top P_1^\top & A_1^\top P_2 \\ P_0 + A_2^\top P_1^\top& P_1+A_2^\top P_2\ear  +   \lba{cc}P_1A_1 & P_0+ P_1A_2 \\ P_2A_1 & P_1^\top+P_2A_2 \ear  \\- \lba{cc} P_1B_0R_0^{-1} B_0^\top P_1^T  & P_1B_0R_0^{-1}B_0^\top P_2\\ P_2B_0R_0^{-1}B_0^\top P_1^T & P_2B_0R_0^{-1}B_0^\top P_2\ear + \lba{cc} Q_0 & 0 \\ 0 & Q_2 \ear =0.
    \end{aligned}
\end{equation}
In order for $K$ to be completely decentralized, we need $R_0^{-1}B_0^\top P_1$ and $R_0^{-1}B_0^\top P_2$ to be diagonal. This is because \[u^*=K\begin{bmatrix}
    x\\\dot{x}
\end{bmatrix}=R_0^{-1}B_0^\top P_1 x+R_0^{-1}B_0^\top P_2\dot{x}\] so if both components of $K$ are diagonal, then for $i=1,\ldots,n$,  $u^*_i=k_ix_i+l_i\dot{x}_i$ for some scalars $k_i,l_i$. Thus, each actuation signal $u_i$ is  a function of only state elements that its
subcontroller $Z_i$ has direct access to and no communication with other subcontrollers is required. 

\subsection{Block circulant case}
Assume that $A_1,A_2,B_0,Q_0,Q_2, R_0$ are all circulant matrices. 
From the (1,1) entry of \eqref{eq:bigARE} we see that  $P_1\succ 0 $ solves the ``smaller" ARE: 
\begin{equation}\label{eq:P_1ARE}
    A_1^\top P_1^{\top} + P_1 A_1 -  P_1 B_0 R_0^{-1}B_0^\top P_1^{\top} + Q_0 = 0.
\end{equation}
Using the result in Theorem \ref{thm:nbynSIcondition}, we can derive a condition that completely decentralizes the system associated with \eqref{eq:P_1ARE}, i.e. conditions for $R_0^{-1}B_0^\top P_1$ to be diagonal.

Next, from the (2,2) entry of \eqref{eq:bigARE}, we know that 
$$
    P_1 + P_1 ^{\top} +A_2^\top P_2 + P_2 A_2 - P_2B_0R_0^{-1}B_0^\top P_2 + Q_2 = 0. 
$$  Define $\overline{Q} = Q_2+P_1+P_1^\top$ and note that $\bar{Q}\succ 0 $ since $Q_2, P_1\succ 0$. Then $P_2\succ 0 $ solves the ``smaller'' ARE 
\begin{equation}\label{eq:P_2ARE}
    A_2^\top P_2 + P_2 A_2 - P_2B_0R_0^{-1}B_0^\top P_2 + \overline{Q} = 0. 
\end{equation} Again using the result in Theorem \ref{thm:nbynSIcondition}, we can derive a condition for diagonal $R_0^{-1}B_0^\top P_2$.
 Thus, we have described a method to derive conditions to completely decentralize controller for a block circulant system by iteratively solving simpler, circulant systems. This is a generalization of the condition for complete decentralization of LQG for the discrete wave equation on the unit circle, which exhibits second order dynamics, described in \cite{mccurdy2025complete}.

\subsection{$2\times2$ blocks with $B,Q,R$ diagonal}
Now, assume that $x(t),u(t)\in\R^2$ so that $A_1,A_2,B_0 ,Q_0,Q_2,$ and $R_0$ in \eqref{eq:2nd_order} are elements of $\R^{2\times 2}$.  Additionally, assume $B_0=I$ and $Q_0,Q_2,R_0$ are diagonal. From the (1,1) entry of \eqref{eq:bigARE} we see that  $P_1\succ 0 $ solves the ``smaller" ARE: 
\begin{equation}\label{eq:P_1ARE2by2}
    A_1^\top P_1^{\top} + P_1 A_1 -  P_1  R_0^{-1} P_1^{\top} + Q_0 = 0.
\end{equation}
Theorem \ref{thm:2x2simplecondition} gives conditions that completely decentralize the LQR associated with \eqref{eq:P_1ARE2by2}, i.e. conditions for $R_0^{-1} P_1$ to be diagonal. Since $R_0^{-1}$ is diagonal, this implies $P_1$ is as well. 

Next, from the (2,2) entry of \eqref{eq:bigARE}, we know that $P_2\succ 0 $ solves the ``smaller'' ARE 
\begin{equation}\label{eq:P_2ARE2by2}
    A_2^\top P_2 + P_2 A_2 - P_2B_0R_0^{-1}B_0^\top P_2 + \overline{Q} = 0. 
\end{equation} where $\overline{Q} = Q_2+P_1+P_1^\top$. Note that $\overline{Q}\succ 0 $ and is diagonal since since $Q_2, P_1\succ 0$ and are diagonal. Again using the result in Theorem \ref{thm:2x2simplecondition}, we can derive a condition for diagonal $R_0^{-1} P_2$, which combining with the previous condition gives a fully decentralized controller. 

\subsection{$2\times2$ block circulant case}
Finally, again assume that $x(t),u(t)\in\R^2$ so that $A_1,A_2,B_0,Q_0,Q_2,R_0\in \R^{2\times 2}$ but are now all circulant. We can repeat the exact same iterative process as in the two previous sections, except using Corollary \ref{cor:2by2circcond} for each of the smaller AREs. Using Corollary \ref{cor:2by2circcond} as opposed to Theorem \ref{thm:nbynSIcondition} allows us to gain more physical insight into the decentralization of the more complex system.

\section{Conclusion}
Although a complete characterization of necessary and sufficient conditions for LQR remains an open question, this paper provided a partial answer. Notably, while previous case studies for this decentralization phenomenon [\cite{jensen2020localization, arbelaiz2024optimal, arbelaiz2022information, mccurdy2025complete}] were restricted to both (1) spatially-invariant dynamics {and} (2) coupled cost functionals, this work demonstrated that decentralization may occur even with decoupled cost functionals and with spatially-varying dynamics. It was demonstrated that a sufficient condition for decentralization with a decoupled cost function in the $2\times 2$ setting was opposite signed off diagonal terms in the state transition matrix, which were interpreted as a ``competitive" dynamic quality. On the other hand it was shown that for a specific case of ``cooperative" dynamics, a coupled cost is required for decentralization. LQR cost analysis for an example system highlighted that performance and decentralization are not intrinsically  competing objectives, suggesting that future co-design procedures could account for both of these behaviors in system design.

Generalizing the results in this paper to find a complete characterization of necessary and sufficient conditions for decentralization of LQR is the subject of ongoing work. Future work aims to qualify decentralization conditions for other optimal control policies (e.g., $H_{\infty}$), to rigorously characterize a relationship between controller locality structure and controller performance, and to qualify the robustness of these results.


\bibliography{paper}    

\appendix
\section{Lemmas} We use the following well-established fact:
\begin{lem}\label{lem:rootsincommon}
    Two quadratics:
$$
    \begin{aligned}
       & x^2 + \beta_1 x + \gamma_1 = 0\\
       & x^2 + \beta_2 x + \gamma_2 = 0
    \end{aligned}
$$
have 2 roots in common if and only if 
$$
   1 =  \frac{\beta_1}{\beta_2} = \frac{\gamma_1}{\gamma_2}.
$$
They have exactly one root in common (at value $\alpha$) if and only if
$$
    \alpha = \frac{\beta_1 \gamma_2 - \beta_2 \gamma_1}{\gamma_1 - \gamma_2} = \frac{\gamma_1 - \gamma_2}{\beta_2 - \beta_1}.
$$
\end{lem}
\section{Proofs}

\subsection{Proof of Theorem \ref{thm:2x2simplecondition}}\label{app:2x2simplecondition}
\begin{pf}
The ARE for this system is given by
\[
\begin{aligned}
0 ={}&
\lba{cc}
a_0 & a_1 \\
a_{-1} & a_2
\ear^{\top}
\lba{cc}
p_0 & 0 \\
0 & p_2
\ear
+
\lba{cc}
p_0 & 0 \\
0 & p_2
\ear
\lba{cc}
a_0 & a_1 \\
a_{-1} & a_2
\ear
\\[0.5em]
&-
\lba{cc}
p_0 & 0 \\
0 & p_2
\ear
\lba{cc}
\gamma_0 & 0 \\
0 & \gamma_2
\ear
\lba{cc}
p_0 & 0 \\
0 & p_2
\ear
+
\lba{cc}
q_0 & 0 \\
0 & q_2
\ear .
\end{aligned}
\]

    which reduces to the three scalar valued equalities: 
    \begin{enumerate}
    \item $2 a_0 p_0 - p_0^2 \gamma_0 + q_0 = 0$
    \item $2a_2 p_2 - p_2^2  \gamma_2 + q_2 = 0$
    \item $a_{-1}p_2 + a_1 p_0 = 0$
\end{enumerate}
From equality 3. we have that 
$$
    p_0 = \frac{- a_{-1} p_2}{a_1}
$$
Substitute this into equality 1. to obtain 
\begin{equation} \label{eq:quadratic1}
\begin{aligned} 
p_2^2
+ \left(
    \frac{2a_0a_1}{\gamma_0a_{-1}}
  \right) p_2
+
\left(
\frac{-q_0a_1^2}{\gamma_0a_{-1}^2}
\right) =0
\end{aligned}
\end{equation}
From equality 2., $p_2$ must also be a root of the quadratic: 
\be \label{eq:quadratic2}
    p_2^2 -2 \frac{ a_2}{\gamma_2}p_2 - \frac{q_2}{\gamma_2} = 0. 
\ee 
Summarizing, we obtain the following conditions for complete decentralization: 
\begin{itemize}
    \item The quadratics \eqref{eq:quadratic1} and \eqref{eq:quadratic2} must have at least one strictly positive, real-valued root in common.
    \item For this common positive root, $p_2$, the value $p_0 = \frac{- a_{-1} p_2}{a_1}$ must be strictly positive.
\end{itemize}

By Lemma \ref{lem:rootsincommon}, equations \eqref{eq:quadratic1} and \eqref{eq:quadratic2} will have both of their roots in common if: 
\begin{equation}
    1= \frac{-\gamma_2a_0a_1}{\gamma_0a_2a_{-1}} = \frac{\gamma_2q_0a_1^2}{\gamma_0q_2a_{-1}^2}
\end{equation} \label{eq:q1zero}
The left hand side of Equation \ref{eq:q1zero} implies that: 
\begin{equation}\label{eq:gammacondition_app}
    \frac{\gamma_0}{\gamma_2}=\frac{-a_0a_1}{a_2a_{-1}}
\end{equation} 
while the right hand side implies that
\begin{equation}\label{eq:qcondition_app}
    \frac{q_0}{q_2} = \frac{-a_0a_{-1}}{a_1a_2}
\end{equation}

Assuming that we have chosen Q and R such that the polynomials have both roots in common, we now must assure that at least one of the roots is positive. The roots of \eqref{eq:quadratic2} are given by: 

\begin{equation}
    \frac{a_2}{\gamma_2} + \sqrt{\frac{a_2^2}{\gamma_2^2}+\frac{q_2}{\gamma_2}} , \space \frac{a_2}{\gamma_2} - \sqrt{\frac{a_2^2}{\gamma_2^2}+\frac{q_2}{\gamma_2}}
\end{equation}

With the usual assumption that Q and R are positive definite, the first of these roots is positive for all values of $a_2$. 

This root, which gives $p_2$ must also result in $p_0=\frac{-q_1-a_{-1}p_2}{a_1} = \frac{-a_{-1}p_2}{a_1} > 0$ (with the assumption that $q_1=0$). Since $p_2>0$, we have that $p_2>0$ as long as $a_1,a_{-1}$ have opposite signs. 

Finally, we must ensure that the values of Equations \eqref{eq:gammacondition_app} and \eqref{eq:qcondition_app} are positive so that Q and R are positive definite. Since $a_1,a_{-1}$ must have opposite signs, this condition is met as long as $a_0,a_2$ have the same sign. 

\end{pf}
\subsection{Proof of Theorem \ref{thm:nbynSIcondition}}\label{app:nbynSIcondition}
\begin{pf}
    The solution $P$ to \eqref{eq:ARE} can be parameterized in the frequency domain by solving the family of scalar-valued ARE's 
\begin{equation}\label{eq:frequencyARE}
     \hat{a}(\kappa) \hat{p}(\kappa) + \hat{p}(\kappa)\hat{a}(\kappa) - \hat{p}(\kappa) \hat{b}(\kappa) \frac{1}{\hat{r}(\kappa)} \hat{b}(\kappa) \hat{p}(\kappa) + \hat{q}(\kappa) = 0.
\end{equation}
The LQR gain $K = R^{-1} B^{\top} P$
can be represented by the Fourier series coefficients of the sequence forming its first row as
\begin{equation}\label{eq:k_frequency}
    \hat{k}(\kappa) = \frac{\hat{b}(\kappa) \hat{p}(\kappa)}{\hat{r}(\kappa)}.
\end{equation}
Thus, $K$ is completely decentralized when  $\exists~ c$ s.t. 
\begin{equation}
    \frac{\hat{b}(\kappa) \hat{p}(\kappa)}{\hat{r}(\kappa)} = c,~~ \forall \kappa.
\end{equation}
    
The scalar-valued equations \eqref{eq:frequencyARE} and \eqref{eq:k_frequency}
can be solved analytically to obtain an expression for $\hat{k}(\kappa)$ as a function of system parameters: 
\begin{equation} \label{eq:solve_for_k}
    \hat{k}(\kappa) = \frac{\hat{a}(\kappa) + \sqrt{\hat{a}(\kappa)^2 + \hat{b}(\kappa)^2 \hat{q}(\kappa)/\hat{r}(\kappa)}}{\hat{b}(\kappa)}
\end{equation}
Then for decentralization we require existence of constant $c$ for which 
\begin{equation}
    \label{eq:quadratic_for_c}
    0 =   c^2 - 2 c \frac{\hat{a}(\kappa)}{\hat{b}(\kappa)} - \frac{\hat{q}(\kappa)}{\hat{r}(\kappa)}
\end{equation}
\end{pf}

\subsection{Proof of Corollary \ref{cor:2by2circcond}}\label{app:2by2circcond}
\begin{pf}
    Matrices of the form \eqref{eq:2by2circ} admit Fourier coefficients
$$ \begin{aligned}
    & \hat{a}(0) = \tfrac{1}{\sqrt{2}} \left(a_0 + a_1 \right)\\
    & \hat{a}(1) =  \tfrac{1}{\sqrt{2}} \left(a_0 - a_1 \right).
\end{aligned} $$

From this form of Fourier coefficients, we can write the family of equations \eqref{eq:SIcondition} for the $2 \times 2$ case as 
\begin{equation}\label{eq:quadfam}\begin{aligned}
  &  0 = c^2 -  {2c}\frac{ \left( a_0 + a_1 \right)}{ \left( b_0 + b_1 \right)} - \frac{\left( q_0 + q_1 \right)}{ \left( r_0 + r_1 \right)},\\
  & 0 = c^2 -  {2c}\frac{ \left( a_0 - a_1 \right)}{ \left( b_0 - b_1 \right)} - \frac{\left( q_0 - q_1 \right)}{ \left( r_0 - r_1 \right)}
\end{aligned} \end{equation}
For decentralization, we require that there exists a single $c$ that satisfies both equations in \eqref{eq:quadfam}. By Lemma \ref{lem:rootsincommon}, \eqref{eq:quadfam} has both roots in common if and only if
\begin{equation}\label{eq:2by2bothroots}
   1 = \frac{\left(\tfrac{a_0+a_1}{b_0+b_1}\right)}{\left(\tfrac{a_0-a_1}{b_0-b_1}\right)} = \frac{\left(\tfrac{q_0+q_1}{r_0+r_1}\right)}{\left(\tfrac{q_0-q_1}{r_0-r_1}\right)} 
\end{equation}

thus \eqref{eq:2by2bothroots} is a sufficient condition for complete decentralization, which gives \eqref{eq:bothroots} when rearranged.
\end{pf}
\end{document}